\documentclass[a4paper,USenglish,autoref,cleveref]{lipics-v2021}

\pdfoutput=1 
\hideLIPIcs  
\nolinenumbers
\title{Temporal Connectivity Augmentation} 

\author{Thomas Bellitto}{Sorbonne University, CNRS, LIP6, F-75005 Paris, France}{thomas.bellitto@lip6.fr}{}{}

\author{Jules Bouton Popper}{Sorbonne University, CNRS, LIP6, F-75005 Paris, France}{jules.bouton-popper@lip6.fr}{}{}

\author{Bruno Escoffier}{Sorbonne University, CNRS, LIP6, F-75005 Paris, France}{bruno.escoffier@lip6.fr}{}{}

\authorrunning{T. Bellitto, J. Bouton Popper and B. Escoffier} 

\Copyright{Thomas Bellitto, Jules Bouton Popper and Bruno Escoffier} 



\category{} 

\relatedversion{} 

\newcommand*\itemif{\item[$\Rightarrow$]}
\newcommand*\itemonlyif{\item[$\Leftarrow$]}
\newcommand{\dintervals}[2]{\{#1,\dotsc,#2\}}

\DeclareMathOperator*{\argmin}{arg\,min}
\usepackage[linesnumbered,lined]{algorithm2e}
\usepackage{tikz}
\usetikzlibrary{positioning}
\usetikzlibrary{fit}
\usetikzlibrary{arrows, automata}
\newtheorem*{theorem*}{Theorem}
\newtheorem*{propertyP}{Property $\mathcal{P}$}
\usepackage{amssymb}

\begin{document}

\maketitle

\begin{abstract}
    Connectivity in temporal graphs relies on the notion of  temporal paths, in which edges follow a chronological order (either strict or non-strict). In this work, we investigate the question of how to make a temporal graph connected. More precisely, we tackle the problem of finding, among a set of proposed temporal edges, the smallest subset such that its addition makes the graph temporally connected (TCA). We study the complexity of this problem and variants, under restricted lifespan of the graph, i.e. the maximum time step in the graph. Our main result on TCA is that for any fixed lifespan at least 2, it is NP-complete in both the strict and non-strict setting. We additionally provide a set of restrictions in the non-strict setting which makes the problem solvable in polynomial time and design an algorithm achieving this complexity. Interestingly, we prove that the source variant (making a given vertex a source in the augmented graph) is as difficult as TCA. On the opposite, we prove that the version where a list of connectivity demands has to be satisfied is solvable in polynomial time, when the size of the list is fixed. Finally, we highlight a variant of the previous case for which even with two pairs the problem is already NP-hard.
\end{abstract}

\section{Introduction}

    Temporal graphs are a generalization of graphs, where edges between two nodes can only be used at given time steps. Adding a notion of temporality on the edges makes temporal graphs a powerful tool to study dynamic systems. For example, temporality makes them much more suitable to model public transport networks, where trains between two cities only run at given times, or to study the spread of information or an epidemic.
    
    Generalizing notions of graph theory to temporal graphs is an important and active topic but also a challenging one. Indeed, we lose very fundamental properties of accessibility in graph such as symmetry and transitivity: if there is an edge between $u$ and $v$ at time 1 and one between $v$ and $w$ at time 2, it is possible to go from $u$ to $w$ but not from $w$ to $u$ even though one can go from $w$ to $v$ and from $v$ to $u$. As such, notions such as connected components or spanning trees do not translate well to temporal graphs.
    
    The main topic of this paper is the question of how to (re)connect a temporal graph. Spanners are a well-studied topic in temporal graph theory \cite{kempe2000connectivity, axiotis2016size, akrida2017complexity, casteigts2021temporal, casteigts2024sharp, angrick2024towards} and can be seen as a generalization of spanning trees. The problem of minimal spanner is, given a connected graph, to remove as many edges as possible without losing the connectivity of the graph. Here, we take a converse approach and start from a graph which is not connected and look for the smallest set of temporal edges to add to make the graph connected. For example, if we want to strengthen a transportation network, which line would be most beneficial for the connectivity of the network. An important variant is to pick the edges we add within a given set of possible edges. For example, if a network is damaged, a solution to this problem would highlight which edges should be repaired first to make the network connected again as quickly as possible.

    Using the model of temporal graphs introduced by Kempe et al. \cite{kempe2000connectivity} where time is discretized into time steps and the edge set can evolve from a time step to the next one, we formalize our problem as follows: given a temporal graph $\mathcal{G}$ and a set $F$ of temporal edges (not in  $\mathcal{G}$), we want to determine a set $F'\subseteq F$ of minimal size such that adding $F'$ to $\mathcal{G}$ restore some connectivity requirements (see Section~\ref{sec:def} for formal definitions).
    
        \subparagraph*{Our contribution.} Our work focuses on the computational complexity of the problems.
        We look at three different connectivity requirements: in the first one, we want to make the graph temporally connected by addition of temporal edges. We denote the problem by TCA, for temporal connectivity augmentation. In the second one, we fix a specific vertex $v$ and we want it to be a source of the temporal graph (we call this problem TSA, for temporal source augmentation). In the last one, we are given a set of pairs $(u_i,v_i)$ of vertices, and we want that a path exists from $u_i$ to $v_i$ for all $i$ (TPCA, temporal pair connectivity augmentation).  
        
        We mainly show that TCA and TSA are NP-hard even when there are only 2 time steps, while TPCA is polynomially solvable when the number of pairs is fixed. We refine the hardness results by showing that  TCA and TSA remain hard under other restrictions. We also provide some additional complementary results, mainly a specific case where TCA is polynomially solvable, and a natural variation of TPCA which turns out to be NP-hard even with a fixed number of pairs. 
        
        Note that most of these results hold both in the strict and in the non-strict settings, corresponding to the case where temporal paths are required to be strict (involving temporal edges with strictly increasing time) or not.\\
        
        The article is organized as follows. Formal definitions and preliminaries are given in Section~\ref{sec:def}. Section~\ref{sec:tca} present results on TCA and TSA, while Section~\ref{sec:tpca} deals with TPCA.

    \subparagraph*{Further related work.} On the topic of temporal connectivity, introduced by Kempe et al. \cite{kempe2000connectivity}, the work of Bhadra and Ferreira \cite{bhadra2003complexity} focuses on extending the definition of connected components to temporal graphs.
    Furthering the study of connected components, Jarry and Lotker \cite{jarry2004connectivity} focus on temporal graphs with geometric properties.
    Connectivity was also investigated in the time windows setting, for example by Kuhn et al. \cite{kuhn2010distributed}. In this work, the graph has to be connected for every time interval of length $\delta$. A connectivity problem studied by Michael et al. \cite{michael2009maintaining} resembles what our work focuses on. They have a swarm of robots and ask that over time, they maintain connectivity by distributed protocols. On the other side, our approach is centralized, does not start with a connected graph necessarily and only aims at connectivity on the full lifespan of the graph. A framework of analysis has been proposed by Casteigs et al. \cite{casteigts2024simple}, establishing relevant restrictions to the graph when studying connectivity and a hierarchy between them for complexity results.
    
    As mentioned earlier, another area of interest in temporal graphs is the study of spanners, i.e. subgraphs that preserve connectivity. It is first introduced by Kempe et al. \cite{kempe2000connectivity}, as an open question about finding a sparse subset of edges in an already connected graph. A negative answer was given by Axiotis and Fotakis \cite{axiotis2016size}, showing that there are minimal connected temporal graphs with $\Theta(n^2)$ edges. Additionally, they investigate weaker connectivity requirements, as the total connectivity is sometimes not relevant in practical applications. Our work also considers such generalizations. 
    
    In \cite{akrida2017complexity}, Akrida et al. restricted the problem to temporal cliques, i.e. temporal graphs with a complete underlying graph, with unique time labels, but did not improve the asymptotic density of spanners. Furthering the work, in 2021, Casteigts et al. \cite{casteigts2021temporal} proved that temporal cliques admit sparse spanners with $O(n \log n)$ edges, and designed an algorithm to find such a spanner. The question of whether a linear spanner exists in all temporal cliques remains open. A study of the problem on random temporal graphs is done by Casteigts et al. in \cite{casteigts2024sharp} to establish sharp thresholds on connectivity properties and the implications for the existence of different kinds of sparse spanners. In \cite{angrick2024towards} Angrick et al. make a considerable step forward by proving the existence of linear spanners in some families of temporal cliques, leaving very constrained cases where the question is still open. Other research on spanners include robustness \cite{bilo2022blackout}, or inefficient networks \cite{christiann2023inefficiently}.

\section{Preliminaries and notation}\label{sec:def}
    This section introduces important definitions and makes some basic observations that contextualize our problem.

    \begin{definition}[temporal graph]
        A temporal graph is a pair $(G,\lambda)$ where $G=(V,E)$ is a (static) graph and $\lambda:E\rightarrow 2^{\mathbb{N}}$. $\lambda(e)$ represents the set of time steps where edge $e$ is present. The lifespan of the graph is $T=\max_{e\in E}\max \lambda(e)$.  $(G,\lambda)$ is {\it simple} if $|\lambda(e)|=1$ for all $e\in E$. 
    \end{definition}    
    A temporal graph is also classically defined using temporal edges. 
    A temporal edge is a pair $(e,t)$ with $e\in \binom{V}{2}$ and $t\in \mathbb{N}^*$. Then the temporal edges of $(G,\lambda)$ are the temporal edges $(e,t)$ where $e\in E$ and $t\in \lambda(e)$.
    
    The snapshot of $(G,\lambda)$ at time $t$ is the (static) graph $G_t=(V,E_t)$ where $E_t$ is the set of edges $e\in E$ with $t\in \lambda(e)$ (i.e., temporal edges at time $t$). We will call {\it snapshot components} at time $t$ the set of connected components of the snapshot at time $t$.
    
    \begin{definition}[journey]
        A \textit{journey} (or \textit{temporal path}) of a temporal graph $\mathcal{G}$ is a sequence  $(u_0,u_1,t_0),(u_1,u_2,t_1),\dotsc,(u_{k-1},u_k,t_{k-1})$ where:
        \begin{itemize}
            \item $(\{u_i,u_{i+1}\},t_i)$ is a temporal edge of $\mathcal{G}$;
            \item $t_i \leq t_{i+1}$ (resp. $t_i < t_{i+1}$) if the journey is \textit{non-strict} (resp. \textit{strict}).
        \end{itemize}
    \end{definition}
    \begin{definition}[temporal connectivity]
        A vertex $v$ is \textit{reachable} from $u$ if there is a journey from $u$ to $v$ (denoted $u\rightarrow v$). If every vertex $v$ is reachable from every other vertex $u$ in $\mathcal{G}$, then $\mathcal{G}$ is \textit{temporally connected}. If all journeys are strict, then the graph is \textit{strictly connected}.
    \end{definition}

    \begin{definition}[augmentation]
        Let $\mathcal G$ be a temporal graph, and $F$ be a set of temporal edges. The {\it augmentation} of $\mathcal G$ by $F$ is the temporal graph  denoted $\mathcal G \uparrow F$ obtained from  $\mathcal G$ by adding the temporal edges in $F$.
        We say that $F$ is a \textit{(strictly) connecting set} for $\mathcal G$ if $\mathcal G \uparrow F$ is (strictly) connected.
    \end{definition}
    
    Now we can formally define the main problems under consideration.

    \begin{quote}
        \textsc{Temporal Connectivity Augmentation} (TCA)\\
        \textbf{Input:} A temporal graph $\mathcal{G}$, a set $F$ of temporal edges and an integer $K$.\\
        \textbf{Question:} Is there set $F'\subseteq F$ of size at most $K$ such that $\mathcal{G}\uparrow F'$ is temporally connected?
    \end{quote}

    As mentioned in introduction, we also consider other connectivity requirements. More precisely, we consider:
    \begin{itemize}
        \item \textsc{Temporal Source Augmentation} (TSA) where, given a specific vertex $v$, we want to make $v$ a source (by augmentation).
        \item \textsc{Temporal Pairs Connectivity Augmentation} (TPCA), which takes as additional input a list of pairs of vertices that need to be connected in the augmented graph. Formally:
    \begin{quote}
        \textsc{Temporal Pairs Connectivity Augmentation} (TPCA)\\
        \textbf{Input:} A temporal graph $\mathcal{G}$, a set $F$ of temporal edges, a set $X = \{(u_1,v_1),\dotsc,(u_p,v_p)\}$ of $p$ connectivity demands and an integer $K$.\\
        \textbf{Question:} Is there a set $F'\subseteq F$ of size at most $K$ such that in $\mathcal{G}\uparrow F'$ there exists a journey from $u_i$ to $v_i$ (for $i=1,\dots,p$)?
    \end{quote}
    \end{itemize}
    
    TCA, TSA and TPCA are defined using non strict journeys. We consider also the strict versions, with strict connectivity requirements, namely Strict TCA, Strict TSA and Strict TPCA. To avoid some possible confusion, we will sometimes explicitely precise that TCA, TSA, TPCA refer to the non strict version (and then write Non Strict TCA,\dots). \\
    
    We will study the complexity of these problems under several possible restrictions. We will consider the {\it simple case} when $\mathcal{G} \uparrow F$ is simple, the case where the lifespan of $\mathcal{G} \uparrow F$ is bounded by $k$ (denoted $k$-TCA, $k$-TSA, ...), and the {\it unrestricted case} where any temporal edge can be added (on the existing set of vertices and within the lifespan of the graph), i.e., $F=\{(\{u,v\},t): u,v\in V, 1\leq t\leq T\}$.

    To conclude this section we  formally state the connection between spanners and TCA mentionned in introduction. The temporal spanner problem asks, given a connected temporal graph and an integer $K$, if we can maintain connectivity by keeping only $K$ temporal edges of the temporal graph.
    \begin{proposition}
        There is a polynomial-time reduction from the temporal spanner problem to the temporal connectivity augmentation problem.
    \end{proposition}
    \begin{proof}
        Let $I=(\mathcal{G},K)$ be an instance of the temporal spanner problem. We consider a temporal graph $\mathcal{G}'$ on the same vertex set as $\mathcal{G}$, with no temporal edge. We consider as set $F$ (for the augmentation in TCA) all the temporal edges of $\mathcal{G}$. Then, making $\mathcal{G}'$ connected by adding at most $K$ edges of $F$ is equivalent to having a spanner of size at most $K$ in $\mathcal{G}$.
    \end{proof}

\section{Temporal connectivity augmentation}\label{sec:tca}
    
    \subsection{Strict TCA}
        The main result of this section is the following theorem, which states that TCA in the strict setting is NP-complete even for temporal graphs of lifespan 2. This remains true even if the temporal graph is simple, and in the unrestricted case. We also show the same results for the source version Strict TSA.
        
        \begin{theorem}\label{thm:2-STCA}
            Strict 2-TCA is NP-complete, even on simple graphs and even in the unrestricted case.
        \end{theorem}
        \begin{proof}
            The problem is trivially in NP. We make a reduction from \textsc{Dominating Set}. In that problem, given a (static) graph $G = \left( V,E \right)$ and a positive integer $K$, we want to determine if there is a dominating set of size at most $K$ in $G$, i.e., a subset $V'\subseteq V$ with $|V'| \leq K$ such that for all $u \in V \backslash V'$ there is a $v \in V'$ for which $\{u,v\} \in E$. 
            
            Let $I = (G,K)$ be an instance of DOMINATING SET. We construct $\mathcal{G}=(G',\lambda)$ a temporal graph (see figure \ref{fig:2TCAS}) as follows:
            \begin{itemize}
                \item $V' = V \cup \{x,y\}$
                \item $E'$ contains edge $\{x,y\}$ and all edges between any two vertices in  $V \cup \{y\}$
                \item $\lambda(\{u,v\}) = \begin{cases}
                    \{2\} & \text{if }\{u,v\} \in E \cup \{\{x,y\}\} \\
                    \{1\} & \text{else}
                \end{cases}$
            \end{itemize}
            As we are in the unrestricted case, any temporal edge can be added. The maximal number of temporal edges that we can add is $K$.
            
            We now show that there exists a dominating set for $\mathcal G$ of size at most $K$ if and only we can make $\mathcal{G}$ temporally connected by adding at most $K$ temporal edges.
            
            Let us first make some remarks. The underlying subgraph of $G'$ excluding $x$ is complete, meaning that this subgraph is already temporally connected (without adding any edge). Moreover, every vertex $u$ of $V$ can reach $y$ at time 1,  hence the journey $u \xrightarrow{1} y \xrightarrow{2} x$ exists, so we have $u  \rightarrow x$ for every $u\in V$. Then, to make $\mathcal{G}$ temporally connected, we have to ensure that $x \rightarrow u$ for all $u\in V$.            

            \begin{itemize}
                \itemif Suppose that $U$ is a dominating set for $G$ of size at most $K$. For each $u\in U$, we add the edge $\{x,u\}$ at time 1. Then each vertex $u\in U$ is reachable from $x$ at time 1. As $U$ is a dominating set, each vertex $v\in V\setminus U$ is adjacent to some vertex $u\in U$. Edge $\{u,v\}$ is by construction present at time 2, so $v$ is reachable from $x$. The addition of these $|U|\leq K$ temporal edges make the graph connected. 
                \itemonlyif Conversely, suppose that adding a set  of at most $K$ temporal edges makes $\mathcal{G}$ temporally connected. Suppose that we add a temporal edge between two vertices $u$ and $v$ different from $x$. If it is at time 1 this edge is useless to make $x$ a source. If it is at time $2$, this edge is only useful is there is an edge say $\{x,u\}$ or $\{x,v\}$, say $\{x,u\}$, at time 1. Then we can replace the addition of $(\{u,v\},2)$ by the addition of $(\{x,v\},1)$. This means that adding (only) temporal edges incident to $x$ at time 1 is dominant. Then clearly $(\{x,y\},1)$ is useless. So, let $U$ be the set of vertices in $V$ for which a temporal edge $(\{x,u\},1)$ has been added. Note that $|U|\leq K$. Take a vertex $v\in V\setminus U$. By construction of $\mathcal{G}$ and definition of $U$, the unique possible path from $x$ to $v$ is through a vertex $u\in U$ at time 1. Then $v$ must be adjacent to $u$ in $V$. $U$ is a dominating set of size at most $K$.                 
 
            \end{itemize}
            This shows that the unrestricted case is NP-complete. From the above, restricting the augmented set to be in $F=\{(\{x,v)\},1):v\in V\}$ shows NP-hardness for the simple case.  
            
            \begin{figure}[!ht]
    \begin{center}
    \scalebox{.7}{
    \begin{tikzpicture}[node distance={20mm}, thick, main/.style = {draw, circle, minimum size=0.8cm},sing/.style={circle, fill, inner sep=1.5pt}]
        \foreach \i in {1,...,6} {
            \coordinate (v\i) at ({360/6 * (\i - 1)}:2);
            \node[sing] (v\i) at (v\i) {};
        }
        
        \node[draw,inner sep=0.7cm, fit=(v1) (v2) (v3) (v4) (v5) (v6)](FIt1) {};
        \node [below left = 0.2cm and 0.2cm] at (FIt1.south east) {\Large $G = (V,E)$};
        \node[main] (x) [left = 2cm of FIt1.west] {$x$};
        \node[main] (y) [above = 2cm of FIt1.north] {$y$};
        \draw[red, bend left=45] (x) to[looseness=1.5] (y);
    
        \foreach \i/\j in {1/2, 1/3, 1/4, 2/5, 3/5, 4/6, 5/6} {
            \draw (v\i) -- (v\j);
        }
        \foreach \i/\j in {1/5,1/6,2/3,2/4,2/6,3/4,3/6,4/5} {
            \draw[blue] (v\i) -- (v\j);
        }
    
        \foreach \i in {1,...,6} {
            \draw[blue, dashed] (v\i) to[out=90+25*\i, in=0] (x);
            \draw[blue] (v\i) to[out=25*\i, in=-90] (y);
        }
        
        \node[draw=none] (n1) [right=1cm of FIt1.east] {};
        \node[draw=none] (n2) [right=0.5cm of n1] {};
        \node[draw=none,red, align=left] [right = 0cm of n2] {$\lambda(e)=\{2\}, e\notin E$};
        \node[draw=none] (b1) [above=0.3cm of n1] {};
        \node[draw=none] (b2) [above=0.3cm of n2] {};
        \node[draw=none,blue, align=left] [right = 0cm of b2] {$\lambda(e)=\{1\}, e\notin E$};
        \node[draw=none] (r1) [below=0.3cm of n1] {};
        \node[draw=none] (r2) [below=0.3cm of n2] {};
        \node[draw=none,black, align=left] [right = 0cm of r2] {$\lambda(e)=\{2\}, e\in E$};
        \draw[red] (n1) -- (n2);
        \draw[blue] (b1) -- (b2);
        \draw (r1) -- (r2);
    
    \end{tikzpicture}
    }

    \end{center}
    \caption{Example of a transformation of a \textsc{Dominating Set} instance into Strict TCA (in dashed, the augmentation set for the simple case).}
    \label{fig:2TCAS}
\end{figure}
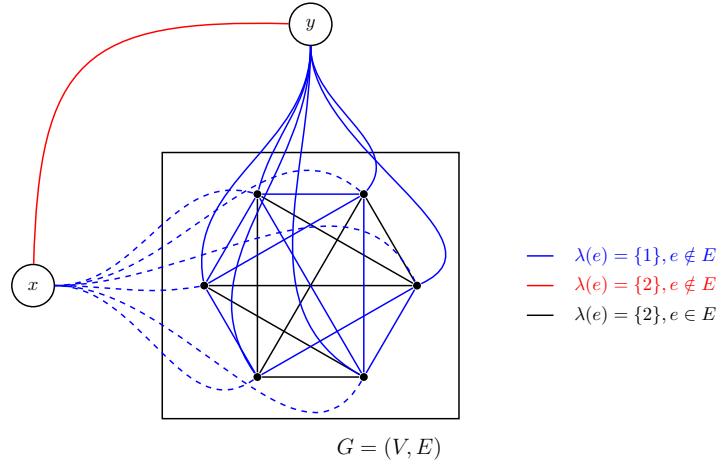
            
        \end{proof}
        In the previous reduction, making the graph temporally connected is exactly the same as making $x$ a source of this temporal graph. Thus we derive the following result:
        \begin{corollary}
            Strict 2-TSA is NP-complete, even on simple temporal graphs and in the unrestricted case.
        \end{corollary}
    
    \subsection{(Non-strict) TCA}
        We focus in this section on the non-strict setting. We  prove that TCA and TSA are NP-complete on simple graphs of lifespan 2. We first introduce a characterization of connected graphs in the non-strict setting which allows for a reformulation of TCA. We then prove the NP-completeness of this reformulation on simple graphs of lifespan 2. Finally we give a separate proof for the NP-completeness of Simple 2-TSA.
        
        We give the following property $\mathcal{P}$, which is verified when there is a "temporal path" from every snapshot connected components at time 1 to every snapshot connected components at the lifespan of the graph:
        \begin{propertyP}\label{prop:P}
            Let $C_1^1,\dotsc,C_{k_1}^1,\dotsc,C_1^T,\dotsc,C_{k_T}^T$ be the snapshot components of a temporal graph $\mathcal G$. $\mathcal G$ satisfies property $\mathcal P$ if for every $i_1 \leq k_1$ and $i_T \leq k_T$, there is a sequence of $i_2,\dotsc,i_{T-1}$ such that $\forall j < T : | C^j_{i_j} \cap C^{j+1}_{i_{j+1}} | \geq 1$.
        \end{propertyP}
        \begin{theorem} \label{thm:char}
            A temporal graph $\mathcal G$ is temporally connected in the non-strict setting if and only if it verifies property $\mathcal P$.
        \end{theorem}
        \begin{proof} We prove that $\mathcal P$ is a characterization of temporally connected graphs:
            \begin{itemize}
                \itemif Suppose $\mathcal G$ is temporally connected. Take $i_1\leq k_1$ and $i_T\leq k_T$, and consider $u\in C^1_{i_1}$ and $v\in C^T_{i_T}$. We know that there exists a journey from $u$ to $v$. Note that by waiting (in $u$ or in $v$) we can assume without loss of generality that the journey starts (in $u$) at time $1$ and ends (in $v$) at time $T$. At each time step $t$, the journey visits some vertices of a connected component $C^t_{i_t}$ (with $C^1_{i_1}$ containing $u$ and $C^T_{i_T}$ containing $v$). The last vertex visited in $C^{t-1}_{i_{t-1}}$ is by definition both in $C^{t-1}_{i_{t-1}}$ and $C^t_{i_t}$, so the intersection is not empty.
    
                \itemonlyif Suppose $\mathcal G$ verifies $\mathcal{P}$, i.e. for every snapshot component at time 1 and at time $T$ we have a sequence of snapshot components starting at 1 and ending at $T$, with a non-empty intersection between consecutive components. For any vertex $u$, we build the journey going to any $v$ with such a sequence. Taking the snapshot component of $u$ at time 1 and the snapshot component of $v$ at time $T$, and then applying the property $\mathcal{P}$, gives us a sequence between those two components. Since any consecutive snapshot component in that sequence has a non-empty intersection, that means that at every step we can travel to any vertex that is part of that intersection in the snapshot component to get to the next one. Repeating this until $v$ is reached gives us a journey from $u$ to $v$.\qedhere
            \end{itemize}
        \end{proof}
        On graphs with a lifespan of 2, we get the following corollary:
        \begin{corollary}\label{cor:t2}
            A graph with a lifespan of 2 is temporally connected in the non-strict setting if and only if for every $i_1 \leq k_1$ and every $i_2 \leq k_2$: $|C^1_{i_1} \cap C^2_{i_2}| \geq 1$.
        \end{corollary}
        In other words, every snapshot component at time 1 has to intersect every snapshot component at time 2. 
        
        We can now reformulate the unrestricted version of 2-TCA as a problem on a binary matrix. Let us consider a temporal graph $\mathcal{G}$ of lifespan 2, with $k_1$ connected components at time 1 and $k_2$ at time 2. We build a matrix $B$, that we call the component-intersection matrix of $\mathcal{G}$, with $k_1$ lines and $k_2$ columns, such that  $B[i,j]=1$ if $C^1_i$ and $C^2_j$ intersect, and 0 otherwise. Adding a temporal edge $(\{u,v\},1)$ in $\mathcal{G}$ corresponds to merging the connected components of $u$ and of $v$ at time 1. On matrix $B$, this corresponds to merging the corresponding raws by an OR-combination. Similarly, adding a temporal edge at time 2 corresponds to merging by an OR-combination the columns corresponding to the components of $u$ and of $v$ at time 2.  
        
        In other words, by the construction above the problem can be reformulated as the following one (where combinations can be on rows on or columns).
        \begin{quote}
            \textsc{OR-Combination To One-Filled} (OCTO)\\
            \textbf{Instance:} A binary matrix $B$ and a positive integer $K$.\\
            \textbf{Question:} Can $B$ be transformed into a one-filled matrix with at most $K$ OR-combinations?
        \end{quote}
        
        Interestingly, the following lemma shows that the other direction works as well.
        
        \begin{lemma}\label{lemmaCC}
        Let $B$ be a binary matrix.
        There exists a simple temporal graph $\mathcal{G}$ of lifespan 2 such that $B$ is its component-intersection matrix.
        \end{lemma}
        \begin{proof}
            We put one vertex $v_{i,j}$ for each $i,j$ such that $B[i,j]=1$. At time 1, we put all edges $\{v_{i,j},v_{i,j'}\}$ for $j\neq j'$, so that row $i$ corresponds to a connected component at time 1. At time 2, we put all edges $\{v_{i,j},v_{i',j}\}$ for $i\neq i'$, so that column $j$ corresponds to a connected component at time 2.     
        \end{proof}
        
        Now we prove that OCTO is NP-complete.
        
        \begin{lemma}\label{lemmaOTO}
            OCTO is NP-complete
        \end{lemma}
        \begin{proof}
            The problem is trivially in NP. We reduce from the decision version of \textsc{Disjoint Set Covers} \cite{cardei2005improving}. In this problem, we are given a finite set $T = \{e_1,\dotsc,e_n\}$, a collection $C = \{T_1,\dotsc,T_m\}$ of subsets of $T$ and a positive integer $K$. We want to determine if $C$ can be partitioned into at least $K$ disjoint parts, such that every part is a cover of $T$. 
 
            Let $I=(T,C,K)$ be an instance of DSC. We construct $B$, a binary matrix with $n (m+1)$ rows and $m$ columns as follows:
            $B[i,j] = \begin{cases}
                1 & e_i \in T_j \\
                0 & e_i \notin T_j
            \end{cases}, \forall i \in \dintervals{1}{n}, \forall j \in \dintervals{1}{m}$. For $i > n$, the matrix repeats itself, $B[i,j] = B[i-n,j]$.
            
            We claim that $I$ is a YES-instance if and only if $I'=(B,m-K)$ is a YES-instance of OCTO.
            \begin{itemize}
                \itemonlyif Suppose $I'$ is a YES-instance of OCTO. Then there is a set of $p \leq m-K$ OR-combinations of rows or columns that results in a one-filled matrix. Note that the combinations are all on columns, since a meaningful combination on rows has to be done in the $m+1$ copies of the matrix (thus with $m>m-K$ combinations). After the OR combinations of columns, each of the $m-p\geq K$ remaining columns corresponds to OR-combinations of a set of initial rows: let $O = \{\{c^1_1,\dotsc, c^1_{k_1}\},\dotsc, \{c^{m-p}_1,\dotsc, c^{m-p}_{k_{m-p}}\}\}$ be the set of combinations done to obtain the columns of the resulting matrix. Since they result from OR combinations, if there is a 1 in some row for some resulting column, that implies that one of the columns used to get that resulting column had a 1 in that row. The matrix is one-filled, so this holds for every resulting column and every row. If we let $\mathcal{C}$ be the same set as $O$ but in which we replace each column by the corresponding set in $C$, we get a disjoint partition of $C$, and each element of the partition is a cover for the rows, thus a cover for $T$. The size is given by $m - p \geq  K$, hence $I$ is a YES-instance of DSC.
    
                \itemif Suppose $I$ is a YES-instance of DSC. Let $P$ be a partition of $C$ into $k \geq K$ part, each part being a disjoint covers of $T$. In the matrix $B$, for every cover, there is a 1 in each row:
                $$ \forall j \in \dintervals{1}{k},\ \forall i \in \dintervals{1}{n(m+1)},\ \exists l \in P_j:\ B[i,l] = 1$$
                Thus combining the columns in each disjoint cover yields a one-filled matrix. This requires $m - k \leq m - K$ combinations, hence $I'$ is a YES-instance of OCTO.\qedhere
            \end{itemize}
        \end{proof}
        Together with Lemma~\ref{lemmaCC}, from Lemma~\ref{lemmaOTO} we derive the following result.  
        \begin{theorem}
            Non-strict 2-TCA is NP-complete, even on simple graphs and even in the unrestricted case.
        \end{theorem}

        We now look at the source variant in the non-strict setting. The previous result cannot be extended easily like in the previous section; instead we give a separate reduction that proves the following:
        \begin{theorem} \label{thm:2-TSA}
            Non-strict 2-TSA is NP-complete, even on simple graphs and even in the unrestricted case.
        \end{theorem}
        \begin{proof}
            The problem is trivially in NP. We make a reduction from \textsc{Hitting Set}.
            \begin{quote}
                \textsc{Hitting Set}\\
                \textbf{Input:} A set $S = \{e_1,\dotsc, e_n\}$, a collection of subsets $C = \{S_1,\dotsc,S_m\}, S_i \subseteq S,\ \forall i \in \{1,\dotsc,m\}$ and a positive integer $K \in \mathbb{N}$ with $ K \leq |S|$.\\
                \textbf{Question:} Is there a subset $S' \subseteq S$ with $|S'| \leq K$ such that $S_i \cap S' \neq \emptyset\ \forall i \in \{1,\dotsc,m\}$?
            \end{quote}
            Let $I=(S,C,K)$ be an instance of \textsc{Hitting Set}. We build a temporal graph  $\mathcal G = (G, \lambda)$ as follows (an illustration of this graph is shown in Figure \ref{fig:2TSA2}). Its vertex set $V$ is composed of:
            \begin{itemize}
                \item A vertex $x$
                \item For every set $S_j \in C$ that contains the element $e_i$, a vertex $e_iS_j$ (this set being referred as $V_{SC}$)
                \item For every set of $C$, a vertex $S_i$ (this set being referred as $V_C$)
            \end{itemize}
            
            At time 1, the only present edges are between vertices of $V_{SC}$ such that they correspond to the same element $e_i$. At time 2, there is an edge between every $e_iS_j \in V_{SC}$ and $S_j \in V_C$.

            The reduction is clearly polynomial time. We now prove that there exists a hitting set of $S$ for $C$ with size at most $K$ if and only if there is a set $F'$ of at most $K$ temporal edges such that  $x$ is a source in $\mathcal G \uparrow F'$.
            \begin{itemize}
                \itemonlyif Suppose that $F'$ with $|F'|\leq K$ is such that $x$ is a source in $\mathcal G \uparrow F'$. We claim that we can build $F''$ containing only edges between $x$ and $V_{SC}$ at time 1, such that $x$ is a source in $\mathcal G \uparrow F''$. Suppose there is a temporal edge $(\{u,v\},1) \in F'$ with $u,v$ different from $x$.
                If $x$ cannot reach neither $u$ nor $v$ at time 1, the edge is useless and can also be removed from $F'$. Otherwise $x$ can reach $u$ or $v$, say $u$, at time 1 without using $(\{u,v\},1)$. Then $(\{u,v\},1)$ may be useful to reach $v$ at time 1, but we can replace $(\{u,v\},1)$ by $(\{x,v\},1)$ with the same effect.
                The same reasoning holds for an edge $(\{u,v\},2)\in F'$. If $x$ cannot reach neither $u$ nor $v$ at 2 or before, the edge is useless. If $x$ can reach $u$ at some time without using $(\{u,v\},2)$, we can replace the edge by $(\{x,v\},1)$. After these replacements, we have added only temporal edges of the form $(\{x,v\},1)$. To conclude, suppose that we have added an edge $(\{x,S_i\},1)$, with $S_i \in V_C$. The only temporal edges incident to $S_i$ are at time 2 (apart from except $(\{x,S_i\},1)$), meaning that we do not gain more reachability from $x$ by arriving at time 1 at $S_i$ than at time 2. Then we can safely replace $(\{x,S_i\},1)$ by some $(\{x,e_jS_i\},1)$ (with $e_j\in S_i$). We build $F''$ by applying the described replacements to $F'$, and $S'$ by taking the corresponding elements of $S$ appearing as endpoints in $F''$. Note that $|S'|\leq K$. To reach a vertex $S_j$, $x$ must use a journey $x \xrightarrow[]{1} e_iS_j \xrightarrow[]{2} S_j$ eventually taking multiple edges in the snapshot component at time 1 containing $e_iS_j$. This implies that $e_i \in S_j$, hence $S_j\cap S'\neq \emptyset$. Since every $S_j$ is reachable from $x$, $S'$ is a hitting set.
                \itemif Conversely, suppose that there exists $S' \subseteq S$, a hitting set with size at most $K$ for $C$. Then consider $F'$ as the set of edges $(\{x,e_iS_j\},1)$ for each $e_i\in S'$, with $S_j$ chosen arbitrarily. Note that $|F'|\leq K$. After augmentation, $x$ can reach every $S_j$ that contains at least one element of $S'$. Since $S'$ is a hitting set, it follows that $x$ can reach all vertices $S_j$. Then at time 2 it can reach all remaining $e_iS_j$ from the vertices of $V_C$. Hence $x$ is a source.
            \end{itemize}
            This shows NP-hardness for the unrestricted case. From the above, restricting the augmented set to be included in $F=\{(\{x,v\},1),v\in V_{SC}\}$ shows NP-hardness for the simple case. 
            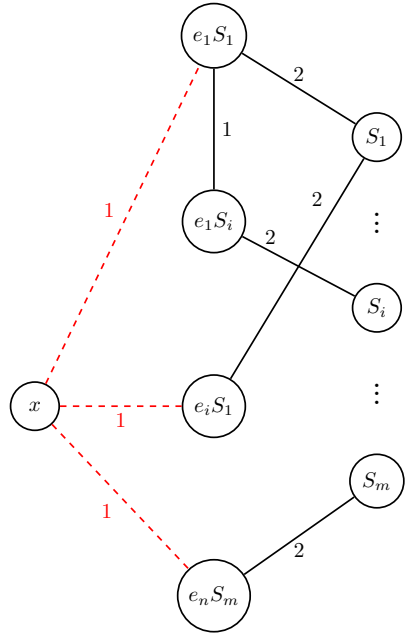
\begin{figure}[!ht]
    \centering
    \scalebox{.8}{
    \begin{tikzpicture}[node distance={20mm}, thick, main/.style = {draw, circle, minimum size=0.8cm}]
        \node[main] (e1S1) [] {$e_1S_1$};
        \node[main] (e1Si) [below = 2cm of e1S1] {$e_1S_i$};
        \node[main] (eiS1) [below = 2cm of e1Si] {$e_iS_1$};
        \node[main] (enSm) [below = 2cm of eiS1] {$e_nS_m$};
        \node[main] (x) [left = 2cm of eiS1] {$x$};
        \node[main] (S1) [below right = 1cm and 2cm of e1S1] {$S_1$};
        \node[main] (Si) [below = 2cm of S1] {$S_i$};
        \node[main] (Sm) [below = 2cm of Si] {$S_m$};
        \node (ints1) [draw=none, below = 0.5cm of S1] {\textbf{$\vdots$}};
        \node (intsi) [draw=none, below = 0.5cm of Si] {\textbf{$\vdots$}};
        \draw[dashed, red] (x) -- node[midway, above left] {1} (e1S1);
        \draw[dashed, red] (x) -- node[midway, below] {1} (eiS1);
        \draw[dashed, red] (x) -- node[midway, below left] {1} (enSm);
        \draw[] (e1S1) -- node[pos=0.5, above] {2} (S1);
        \draw[] (e1Si) -- node[pos=0.25, above] {2} (Si);
        \draw[] (eiS1) -- node[pos=0.75, above left] {2} (S1);
        \draw[] (enSm) -- node[pos=0.5, below] {2} (Sm);
        \draw[] (e1S1) -- node[pos=0.5, right] {1} (e1Si);
    \end{tikzpicture}
    }
    \caption{Example of a transformation of an HS instance into 2-TSA. Dashed red edges correspond to the augmentation set for the simple case.}
    \label{fig:2TSA2}
\end{figure}
        \end{proof}
    
    \subsection{(1+1)-TCA}

        We consider in this section the following situation: suppose that we are given a static (non connected) graph, and that are given an extra time slot, $T=2$, where we can add temporal edges to make the graph connected. The question is to find the minimal number of edges that must be added for this. This is a particular case of TCA (where $F$ is the set of temporal edges at time $2$). We show in this section that it is polynomially solvable.
 
        We begin by the following property on connected graphs of lifespan 2. 
        \begin{lemma} \label{thm:comp}
             Let $\mathcal{G}$ be a temporally connected graph of lifespan 2. Let also $C^1_1,\dotsc,C^1_{k_1}$ and $C^2_1,\dotsc,C^2_{k_2}$ be the snapshot components at time 1 and 2. We have $k_1 \leq \min_{i_2 \in \dintervals{1}{k_2}}|C^2_{i_2}|$ and $k_2 \leq \min_{i_1 \in \dintervals{1}{k_1}}|C^1_{i_1}|$.
        \end{lemma}
        \begin{proof}
            The proof is done by contradiction.
        \end{proof}
        This gives the core idea for a polynomial time algorithm. We construct at time 2 stars centered at each vertex of the smallest component of the first time step. Every star has at least a branch to all the other components of the fixed time step. The remaining vertices are arbitrarily linked to the first star, see \ref{algo:onebeat} for a pseudocode. 
        We give a formal proof of optimality:
        
        \begin{theorem}
             $(1+1)$-TCA is solvable in time $O(n^2)$.
        \end{theorem}
        \begin{proof}
            The algorithm  creates $|C^1_\ell|$ stars that correspond to connected components at time 2. By construction, each star intersects all connected components at time 1, hence by Corollary~\ref{cor:t2} the graph is temporally connected. At time 2, each connected component $C^2_{i_2}$ has $|C^2_{i_2}| - 1$ edges (since it is a star), which sums to $\sum_{i_2 = 1}^{|C^1_\ell|} (|C^2_{i_2}| - 1) = n - |C^1_\ell|$. Suppose we add $p$ links at time 2, we have $k_2 \geq n - p$. To be connected we know $k_2 \geq |C^1_\ell|$, where $C^1_\ell$ is the smallest component at time 1, hence $p \geq n - |C^1_\ell|$.
        \end{proof}
        
        \begin{figure}[!ht]
    \begin{center}
        \begin{algorithm}[H]
            \SetKwInOut{Input}{Input}\SetKwInOut{Output}{Output}\SetKw{Return}{return}
            \Input{
            \begin{itemize}\item A temporal graph $\mathcal{G}$ of lifespan 1 
            \end{itemize}}
            \Output{\begin{itemize}
                \item A set $F'$ of temporal edges at time 2 such that $\mathcal{G}\uparrow F'$ is temporally connected.
            \end{itemize}}
            $F' \leftarrow \emptyset$ \\
            $\ell \leftarrow \argmin_{i \in \dintervals{1}{k_1}} |C^1_i|$ \\
            \For{$i \in \dintervals{1}{k_1}\backslash \{\ell\}$}{
                \For{$j \in \dintervals{1}{|C^1_i|}$}{
                    $s \leftarrow (j \mod \ell) + 1$ \\ 
                    $F' \leftarrow F' \cup \left(\{u^{C^1_i}_j,u^{C^1_\ell}_s \},2\right)$
                }
            }
            \Return{$\mu$}
            \caption{$(1+1)$-TCA algorithm \label{algo:onebeat}}
        \end{algorithm}
    \end{center}
\end{figure}

\section{Temporal pair connectivity augmentation }\label{sec:tpca}
    A natural extension to consider is the problem where instead of all the pairs, we are given a set of pairs to connect by augmentation. As a generalization of TCA and TSA, TPCA is NP-hard both in strict and non-strict settings. We tackle here the problem when the number of pairs is fixed.

    \subsection{Strict TPCA for a fixed number of pairs}
        This problem echoes the \textsc{Directed Generalized Steiner} network problem \cite{charikar1999approximation}, in static graphs:
        \begin{quote}
            \textsc{DG-Steiner}\\
            \textbf{Input:} A directed graph $G = (V,A)$, a weight function $w: A \rightarrow \mathbb R$, a set $X = \{(u_1,v_1),\dotsc,(u_p,v_p)\}$ of $p$ node pairs and two positive integers $B \leq p$ and $K \leq w(\mathcal{G})$.\\
            \textbf{Question:} Is there a subgraph $H$ of $G$ of cost $w(H) \leq B$ that satisfies at least $K$ node pairs from the set $X$, when a pair $(u_i,v_i)$ is said to be satisfied when there is a path from $u_i$ to $v_i$?
        \end{quote}
        We will show they are in fact closely related. We employ a similar definition, adapted to temporal graphs:
        \begin{quote}
            \textsc{TG-Steiner}\\
            \textbf{Input:}A temporal graph, $\mathcal{G}$, a weight function $w$ on the temporal edges of $\mathcal{G}$, a set $X = \{(u_1,v_1),\dotsc,(u_p,v_p)\}$ of $p$ node pairs and two positive integers $K \leq p$ and $B \leq w(\mathcal{G})$.\\
            \textbf{Question:} Is there a subgraph $\mathcal H$ of $\mathcal G$ of cost $w(\mathcal H) \leq K$ that satisfies at least $B$ node pairs from the set $X$?
        \end{quote}
        This definition generalizes TPCA: When $B=p$ and the weight assigned to the edges are 0 for edges initially present in the temporal graph and 1 for the edges proposed for augmentation, the problem becomes equivalent. 
        
        We now show that the \textsc{TG-Steiner} problem for temporal graphs reduces to \textsc{DG-Steiner}:
        \begin{lemma} \label{thm:TGDG}
            Any instance $I$ of the \textsc{TG-Steiner} problem can be transformed into an instance $I'$ of the \textsc{DG-Steiner} problem in polynomial time.
        \end{lemma}
        \begin{proof}
            Let $I=(\mathcal{G}, w, X, K, B)$ be an instance of \textsc{TG-Steiner}. Temporal expansion is a classical technique to transform a temporal graph into a static one~\cite{michail2016introduction}. We build here a specific expansion that fits our augmentation problem. Let $G_T = (V_T,E_T)$ be (static directed) the graph such that:
            \begin{itemize}
                \item For each vertex $v$ of $\mathcal{G}$ we have $T+1$ copies $v^t$, $t=1,\dots,T+1$ in $V_T$. We also have $T$ arcs $(v^t,v^{t+1})$, $i=1,\dots,T$, of weight 0.
                \item For each temporal edge $(\{u,v\},t)$ of weight $w$ of $\mathcal{G}$ we add two vertices $(\{u,v\},t)$ and $(\{u,v\},t)'$, with an arc of weight $w$ from the former to the latter. We also add 4 arcs of weight 0: from $u^t$ and from $v^t$ to $(\{u,v\},t)$, and from $(\{u,v\},t)'$ to $u^{t+1}$ and to $v^{t+1}$.  
            \end{itemize}
The construction is illustrated in Figure \ref{fig:expnonstric}. For the set of pairs, if $(u_i,v_i) \in X$ then $(u_{i}^1,v_{i}^T) \in X'$, i.e. we replace the source by its earliest copy and the sink by its latest copy. We claim that finding a subgraph $H_T$ of $G_T$ of cost $w_T(H_T) \leq B$ for which at least $K$ node pairs from the set $X'$ are satisfied is equivalent to solving \textsc{TG-Steiner} on $I$.
            \begin{itemize}
                \itemif Suppose there exists such a subgraph $H_T$. Note that the only arcs that have a positive weight are those from $(\{u,v\},t)$ to $(\{u,v\},t)'$. We take in $\mathcal{G}$ all temporal edges $(\{u,v\},t)$ such that $((\{u,v\},t),(\{u,v\},t)')$ is in $H_T$.
                In this subgraph of $\mathcal{G}$, which has weight at most $K$, the connected pairs are the same as those connected in $H_T$, by construction of our temporal expansion - a path from $u_{i}^1$ to $b_{i}^T$ in $H_T$ corresponds to a journey from $u_i$ to $v_i$ in $\mathcal{G}'$. Hence, we are guaranteed to satisfy the same pairs in $H_T$ and $\mathcal{G}'$.
            
                \itemonlyif Suppose $I$ is a YES-instance, and let $\mathcal{H}$ be a solution. To construct the solution to the \textsc{DG-Steiner} corresponding instance, the previous process is reversed. We take in  $H_T$ all the arcs of weight 0, and arcs $((\{u,v\},t),(\{u,v\},t)')$ for which the  temporal edges $(\{u,v\},t)$ is selected in $\mathcal{H}$. For the same reasons as before, $H_T$ has the same weight as $\mathcal{H}$ and satisfies the same pairs.\qedhere
            \end{itemize}
        \end{proof}
        
        The previous lemma allows us to use the result of \cite{feldman2006directed}:
        \begin{theorem*}[Feldman, Ruhl 2006]
            \textsc{$p$-DG-Steiner} is polynomial time solvable when $p$ is a constant. There exists an algorithm solving \textsc{$p$-DG-Steiner} in  $O(mn^{4p-2}+n^{4p-1}\log n)$.
        \end{theorem*}
        By noticing the temporal expansion of a temporal graph has $O(nT + M)$ vertices and $O(nT + M)$ arcs, where $M$ denotes the number of temporal edges of the graph, we get through the reduction an algorithm running in time $O((nT+M)^{4p-1}\log(nT+M))$. Hence:
        \begin{theorem}
            Strict TPCA is solvable in polynomial time when parameterized by $p$, the number of pairs to connect. 
        \end{theorem}

    \subsection{Non-strict TPCA for a fixed number of pairs}
        Similarly to the strict case, the temporal expansion allows for the use of the result of \cite{feldman2006directed}. The key observation is their result works for any directed graph (not only on DAGs). 
        It is then possible to design a temporal expansion adapted for the non-strict case, taking into account the augmentation part. We achieve this by the following simple addition of arcs: for each pairs of adjacent temporal edges $(\{u,v\},t)$ and $(\{u,w\},t)$ (at the same time step), we add the arcs $((\{u,v\},t)',(\{u,w\},t))$ and $((\{u,w\},t)',(\{u,v\},t))$ in the expansion, see the gray arcs in figure \ref{fig:expnonstric}. On one hand, this allows paths to travel through multiple edges sharing an endpoint at the same time step. On the other hand, it is impossible to go back in time. As a result, the property that a path in the expansion corresponds to a path in the original graph is preserved. Therefore, we have:
        \begin{theorem}
            Non-strict TPCA is solvable in polynomial time when parameterized by $p$, the number of pairs to connect.
        \end{theorem}
    \begin{figure}[!ht]
    \centering
    \begin{tikzpicture}[node distance={25mm}, thick, main/.style = {draw, circle, minimum size=0.8cm}]
        \node[main] (1) {$u_1$};
        \node[main] (2) [right of=1] {$u_2$};
        \node[main] (3) [below of=2] {$u_3$};
        \node[main] (4) [left of=3] {$u_4$};
        
        \draw[] (1) -- node[pos=0.75, above, text width = 20mm] {$t=\{1,2\}$\ $w(1)=0$\ $w(2)=0$} (2);
        \draw[] (2) -- node[midway, right, text width = 20mm] {$t=\{1\}$\ $w(1)=1$} (3);
        \draw[] (3) -- node[pos=0.25, below, text width = 20mm] {$t=\{2\}$\ $w(2)=0$} (4);
        
        \node[main] (11) [above right = 6mm and 20mm of 2] {$u^1_1$};
        \node[main] (21) [right of = 11] {$u^1_2$};
        \node[main] (31) [right of = 21] {$u^1_3$};
        \node[main] (41) [right of = 31] {$u^1_4$};
        \node[main] (12) [below of = 11] {$u^2_1$};
        \node[main] (22) [right of = 12] {$u^2_2$};
        \node[main] (32) [right of = 22] {$u^2_3$};
        \node[main] (42) [right of = 32] {$u^2_4$};
        \node[main] (13) [below of = 12] {$u^3_1$};
        \node[main] (23) [right of = 13] {$u^3_2$};
        \node[main] (33) [right of = 23] {$u^3_3$};
        \node[main] (43) [right of = 33] {$u^3_4$};
        \node[] (equiv) [above left = 3mm and 4.5mm of 12] {\scalebox{2}{$\Leftrightarrow$}};
        
        \node[draw, circle] (i1121) [below right = 0.35cm and 8.1mm of 11] {};
        \node[draw, circle] (o1222) [above right = 0.35cm and 8.1mm of 12] {};

        \node[draw, circle] (i2131) [below right = 0.35cm and 8.1mm of 21] {};
        \node[draw, circle] (o2232) [above right = 0.35cm and 8.1mm of 22] {};
        
        \node[draw, circle] (i1222) [below right = 0.35cm and 8.1mm of 12] {};
        \node[draw, circle] (o1323) [above right = 0.35cm and 8.1mm of 13] {};
        
        \node[draw, circle] (i3242) [below right = 0.35cm and 8.1mm of 32] {};
        \node[draw, circle] (o3343) [above right = 0.35cm and 8.1mm of 33] {};

        \draw[->] (11) -- (i1121);
        \draw[->] (i1121) -- node[midway, left] {0} (o1222);
        \draw[->] (o1222) -- (22);
        \draw[->] (21) -- (i1121);
        \draw[->] (o1222) -- (12);
        
        \draw[->] (21) -- (i2131);
        \draw[->] (i2131) -- node[midway, right] {1} (o2232);
        \draw[->] (o2232) -- (32);
        \draw[->] (31) -- (i2131);
        \draw[->] (o2232) -- (22);
        
        \draw[->] (12) -- (i1222);
        \draw[->] (i1222) -- node[midway, left] {0} (o1323);
        \draw[->] (o1323) -- (23);
        \draw[->] (22) -- (i1222);
        \draw[->] (o1323) -- (13);
        
        \draw[->] (32) -- (i3242);
        \draw[->] (i3242) -- node[midway, right] {0} (o3343);
        \draw[->] (o3343) -- (43);
        \draw[->] (42) -- (i3242);
        \draw[->] (o3343) -- (33);
        
        \draw[->, gray] (o2232) -- (i1121);
        \draw[->, gray] (o1222) -- (i2131);
        
        \draw[->] (11) -- (12); \draw[->] (12) -- (13);
        \draw[->] (21) -- (22); \draw[->] (22) -- (23);
        \draw[->] (31) -- (32); \draw[->] (32) -- (33);
        \draw[->] (41) -- (42); \draw[->] (42) -- (43);
    \end{tikzpicture}
    \caption{A TG-Steiner instance and its equivalent DG-Steiner expansion. The backward arcs in gray correspond to the non-strict version and the unlabelled arcs have weight 0.}
    \label{fig:expnonstric}
\end{figure}
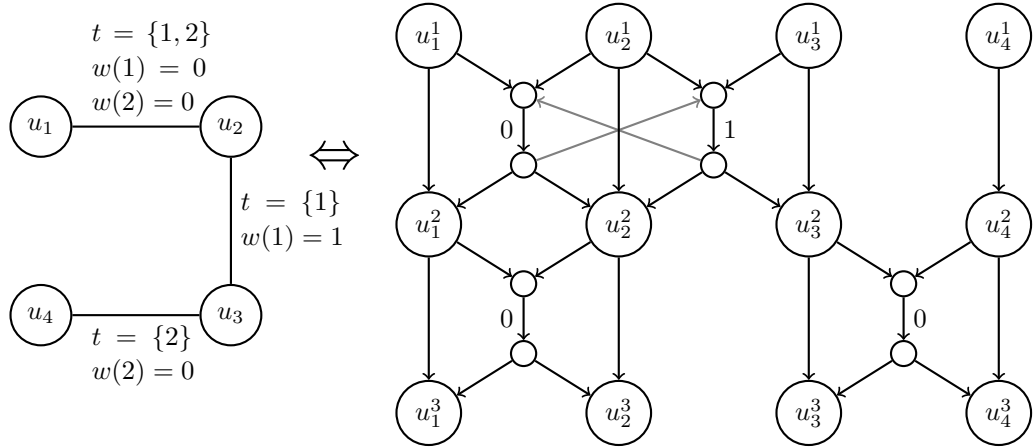
    
    \subsection{Edge-by-edge TPCA}

        Finally, we study a variation of the problem that we call \textit{edge-by-edge} augmentation. Here, we are given a non-connected temporal graph and a set of possible temporal edges that we may add, but the difference is that if several of these edges have the same endpoints (like $(\{u,v\},t)$ and $(\{u,v\},t')$), we can add all of them at once with a cost of only one.
        
        For example this is what happens if we want to strengthen a transportation network and decide to build a new subway line. By building one line (edge), we add many trains (temporal edges) between its two endpoints. 
        
        Since it was the only polynomial case, we first study the TPCA problem under this new setting. We prove that it becomes NP-complete, thus underlying an interesting difference between the two different cost evaluations of an augmentation.
            \begin{theorem}
                \textit{Non-strict edge-by-edge TPCA} is NP-complete even for $p = 2$ and for $T = 2$.
            \end{theorem}
            \begin{proof}
                The problem is trivially in NP. We make a reduction from 3-SAT.\\
                Let $x_1, \dotsc, x_n$ be the variables and $C_1, \dotsc, C_m$ the clauses of the statement $\phi$. We construct our transformation graph in two parts. \\
                In the variable part of the graph, we concatenate gadgets that encode the truth value of the variable, and all the edges are at time 1. A gadget has a starting vertex, connected to two other vertices, that we call the true and false buffer vertices. In each branch, we alternate between value vertices and buffer vertices. There is a value vertex in each branch for each clause where the variable appears. The value vertex is connected to the last buffer vertex by an optional link, present at times $\{1,2\}$. The two value vertices are connected to the two following buffer vertices when they exist. The last two value vertices of each branch are connected to the ending vertex of the gadget. The concatenation is done by connecting the ending vertex to the starting vertex of the next variable.\\
                The clause part of the graph has all of its edges at time 2. Like the variable part, the clauses are encoded by a concatenation of gadgets, composed of a starting vertex and an ending vertex. There is an edge for each literal in the clause, from the starting vertex to the buffer preceding the value vertex corresponding to the clause and in the branch corresponding to the literal. There is also an edge from the corresponding value vertex and the ending vertex in the clause part. An example of variable and clause gadgets is shown in figure \ref{fig:gadgets}.\\
                We want to connect $x_1^S$ to $x_n^E$ and $C_1^S$ to $C_m^E$. Remark that a path from $x_1^S$ to $x_n^E$ only takes edges at time 1 and is always of cost $3m$. This is because the only edges with cost 1 are those added for each literal in each clause and there are 3 literals per clause. Moreover, for a valid solution of SAT and in a variable gadget, only the edges of one branch are added, implying that exactly $3m$ edges have been added. They connect $x_1^S$ to $x_n^E$ and $C_1^S$ to $C_m^E$. Conversely, if there are no solutions with cost $3m$ then the formula cannot be satisfied.
                \begin{figure}[!ht]
    \centering
    \scalebox{.7}{
    \begin{tikzpicture}[node distance={25mm}, thick, main/.style = {draw, circle, minimum size=0.8cm}]
        \node[main] (v_start) {$x_i^S$};
        \node[main] (t_b1) [above right of=v_start] {$\top_b^{C_j}$};
        \node[main] (f_b1) [below right of=v_start] {$\bot_b^{C_j}$};
        \node[main] (t_v1) [right of=t_b1] {$\top_v^{C_j}$};
        \node[main] (f_v1) [right of=f_b1] {$\bot_v^{C_j}$};
        \node[main] (t_b2) [right of=t_v1] {$\top_b^{C_k}$};
        \node[main] (f_b2) [right of=f_v1] {$\bot_b^{C_k}$};
        \node[main] (t_v2) [right of=t_b2] {$\top_v^{C_k}$};
        \node[main] (f_v2) [right of=f_b2] {$\bot_v^{C_k}$};
        \node[main] (v_end) [below right of=t_v2] {$x_i^E$};
    
        \draw[] (v_start) -- node[midway, above left] {1} (t_b1);
        \draw[] (v_start) -- node[midway, below left] {1} (f_b1);
        \draw[] (t_v1) -- node[midway, above] {1} (t_b2);
        \draw[] (f_v1) -- node[midway, below] {1} (f_b2);
        \draw[] (t_v2) -- node[midway, above right] {1} (v_end);
        \draw[] (f_v2) -- node[midway, below right] {1} (v_end);
    
        \node[main] (c1_start) [below = 30mm of v_start] {$C_j^S$};
        \node[main] (c1_end) [right of= c1_start] {$C_j^E$};
        \node[main] (c2_start) [right = 30mm of c1_start] {$C_k^S$};
        \node[main] (c2_end) [right of= c2_start] {$C_k^E$};

        \draw[] (c1_end) -- node[midway, above] {2} (c2_start);
        \draw[] (c1_start) -- node[pos=0.25, above left] {2} (t_b1);
        \draw[] (c1_end) -- node[pos=0.75, right] {2} (t_v1);
        \draw[] (c2_start) -- node[pos=0.25, above left] {2} (f_b2);
        \draw[] (c2_end) -- node[pos=0.25,below right] {2} (f_v2);
    
        \draw[dotted] (t_b1) -- node[midway, above] {1,2} (t_v1);
        \draw[dotted] (t_b2) -- node[midway, above] {1,2} (t_v2);
        \draw[dotted] (f_b1) -- node[pos=0.25, above] {1,2} (f_v1);
        \draw[dotted] (f_b2) -- node[midway, above] {1,2} (f_v2);
        
    \end{tikzpicture}
    }
    \caption{Variable and clause gadgets}
    \label{fig:gadgets}
\end{figure}
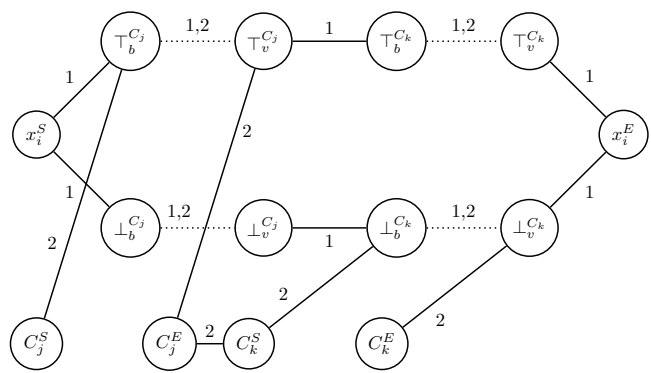
            \end{proof}
            The construction can be adapted for the strict case, however we lose the bound on $T$. Roughly speaking, the path in the variable part is changed from being at time 1 to being a strictly increasing one ending at $\alpha$. Then in the clause part, the path at time 2 is changed to a strictly increasing path starting at $\alpha + 1$. This way we preserve the two distinct paths and the proof still holds:
            \begin{theorem}
                \textit{Strict edge-by-edge TPCA} is NP-complete even for $p = 2$.
            \end{theorem}
            For 2-TCA in the strict setting, we made the observation that only the edges in the augmentation set are interesting to add to the graph. Considering that adding an edge at time 2 is dominated by adding the same edge at time 1, we get that even if we were allowed to add edge-by-edge, the same observation holds. Thus, we get:
            \begin{observation}
                Strict EBE 2-TCA is NP-complete.
            \end{observation}
            As for the non-strict setting, we already establish that the proof is in the unrestricted case by construction. Edge-by-edge would then correspond to allowing doing a column and line merge as one operation, but the optimal solution only does column merges so we would not gain any advantage. It follows then that:
            \begin{observation}
                Non-strict EBE 2-TCA is NP-complete. 
            \end{observation}

\section{Conclusion}
    In this paper, we have introduced and analyzed the Temporal Connectivity Augmentation (TCA) problem and several variants. Our results demonstrate that TCA, in both the strict and non-strict settings, is NP-complete under various constraints, including bounded lifespan and simple graphs.
    
    In addition to these hardness results, we presented polynomial-time algorithms for special cases, such as (1+1)-TCA and TPCA parameterized by the number of pairs. By adapting temporal expansion, we have extended existing techniques to the augmentation setting, and possibly for some weighted temporal problems.
    
    Several promising research directions emerge from our work. A natural extension would involve studying the impact of further structural graph parameters such as those in the work of Enright et al. \cite{enright2024structural}. Additionally, exploring the problem with restrictions on the underlying graph is also an interesting perspective.


\bibliography{refs}

\begin{thebibliography}{10}

\bibitem{akrida2017complexity}
Eleni~C Akrida, Leszek G{\k{a}}sieniec, George~B Mertzios, and Paul~G Spirakis.
\newblock The complexity of optimal design of temporally connected graphs.
\newblock {\em Theory of Computing Systems}, 61:907--944, 2017.

\bibitem{angrick2024towards}
Sebastian Angrick, Ben Bals, Tobias Friedrich, Hans Gawendowicz, Niko Hastrich,
  Nicolas Klodt, Pascal Lenzner, Jonas Schmidt, George Skretas, and Armin
  Wells.
\newblock Towards linear spanners in all temporal cliques.
\newblock {\em arXiv preprint arXiv:2402.13624}, 2024.

\bibitem{axiotis2016size}
Kyriakos Axiotis and Dimitris Fotakis.
\newblock On the size and the approximability of minimum temporally connected
  subgraphs.
\newblock In Ioannis Chatzigiannakis, Michael Mitzenmacher, Yuval Rabani, and
  Davide Sangiorgi, editors, {\em 43rd International Colloquium on Automata,
  Languages, and Programming, {ICALP} 2016, July 11-15, 2016, Rome, Italy},
  volume~55 of {\em LIPIcs}, pages 149:1--149:14. Schloss Dagstuhl -
  Leibniz-Zentrum f{\"{u}}r Informatik, 2016.
\newblock URL: \url{https://doi.org/10.4230/LIPIcs.ICALP.2016.149}, \href
  {https://doi.org/10.4230/LIPICS.ICALP.2016.149}
  {\path{doi:10.4230/LIPICS.ICALP.2016.149}}.

\bibitem{bhadra2003complexity}
Sandeep Bhadra and Afonso Ferreira.
\newblock Complexity of connected components in evolving graphs and the
  computation of multicast trees in dynamic networks.
\newblock In {\em Ad-Hoc, Mobile, and Wireless Networks: Second International
  Conference, ADHOC-NOW2003, Montreal, Canada, October 8-10, 2003. Proceedings
  2}, pages 259--270. Springer, 2003.

\bibitem{bilo2022blackout}
Davide Bil{\`o}, Gianlorenzo D’Angelo, Luciano Gual{\`a}, Stefano Leucci, and
  Mirko Rossi.
\newblock Blackout-tolerant temporal spanners.
\newblock In {\em International Symposium on Algorithms and Experiments for
  Wireless Sensor Networks}, pages 31--44. Springer, 2022.

\bibitem{cardei2005improving}
Mihaela Cardei and Ding-Zhu Du.
\newblock Improving wireless sensor network lifetime through power aware
  organization.
\newblock {\em Wireless networks}, 11:333--340, 2005.

\bibitem{casteigts2024simple}
Arnaud Casteigts, Timoth{\'e}e Corsini, and Writika Sarkar.
\newblock Simple, strict, proper, happy: A study of reachability in temporal
  graphs.
\newblock {\em Theoretical Computer Science}, page 114434, 2024.

\bibitem{casteigts2021temporal}
Arnaud Casteigts, Joseph~G Peters, and Jason Schoeters.
\newblock Temporal cliques admit sparse spanners.
\newblock {\em Journal of Computer and System Sciences}, 121:1--17, 2021.

\bibitem{casteigts2024sharp}
Arnaud Casteigts, Michael Raskin, Malte Renken, and Viktor Zamaraev.
\newblock Sharp thresholds in random simple temporal graphs.
\newblock {\em SIAM Journal on Computing}, 53(2):346--388, 2024.

\bibitem{charikar1999approximation}
Moses Charikar, Chandra Chekuri, To-Yat Cheung, Zuo Dai, Ashish Goel, Sudipto
  Guha, and Ming Li.
\newblock Approximation algorithms for directed steiner problems.
\newblock {\em Journal of Algorithms}, 33(1):73--91, 1999.

\bibitem{christiann2023inefficiently}
Esteban Christiann, Eric Sanlaville, and Jason Schoeters.
\newblock On inefficiently connecting temporal networks.
\newblock In Arnaud Casteigts and Fabian Kuhn, editors, {\em 3rd Symposium on
  Algorithmic Foundations of Dynamic Networks, {SAND} 2024}, volume 292 of {\em
  LIPIcs}, pages 8:1--8:19. Schloss Dagstuhl - Leibniz-Zentrum f{\"{u}}r
  Informatik, 2024.

\bibitem{enright2024structural}
Jessica~A. Enright, Samuel~D. Hand, Laura Larios{-}Jones, and Kitty Meeks.
\newblock Structural parameters for dense temporal graphs.
\newblock In {\em 49th International Symposium on Mathematical Foundations of
  Computer Science, {MFCS} 2024, August 26-30, 2024, Bratislava, Slovakia},
  pages 52:1--52:15, 2024.

\bibitem{feldman2006directed}
Jon Feldman and Matthias Ruhl.
\newblock The directed steiner network problem is tractable for a constant
  number of terminals.
\newblock {\em SIAM Journal on Computing}, 36(2):543--561, 2006.

\bibitem{jarry2004connectivity}
Aubin Jarry and Zvi Lotker.
\newblock Connectivity in evolving graph with geometric properties.
\newblock In {\em Proceedings of the 2004 joint workshop on Foundations of
  mobile computing}, pages 24--30, 2004.

\bibitem{kempe2000connectivity}
David Kempe, Jon Kleinberg, and Amit Kumar.
\newblock Connectivity and inference problems for temporal networks.
\newblock In {\em Proceedings of the thirty-second annual ACM symposium on
  Theory of computing}, pages 504--513, 2000.

\bibitem{kuhn2010distributed}
Fabian Kuhn, Nancy Lynch, and Rotem Oshman.
\newblock Distributed computation in dynamic networks.
\newblock In {\em Proceedings of the forty-second ACM symposium on Theory of
  computing}, pages 513--522, 2010.

\bibitem{michael2009maintaining}
Nathan Michael, Michael~M Zavlanos, Vijay Kumar, and George~J Pappas.
\newblock Maintaining connectivity in mobile robot networks.
\newblock In {\em Experimental Robotics: The Eleventh International Symposium},
  pages 117--126. Springer, 2009.

\bibitem{michail2016introduction}
Othon Michail.
\newblock An introduction to temporal graphs: An algorithmic perspective.
\newblock {\em Internet Mathematics}, 12(4):239--280, 2016.

\end{thebibliography}

\end{document}